\documentclass[a4paper, 12pt]{article}
\usepackage{amsmath,amsthm,amssymb}
\usepackage{graphicx}
\usepackage{color}
\usepackage{dsfont}
\usepackage{CJKutf8}
\usepackage{epstopdf}
\usepackage[hang]{subfigure}
\usepackage{natbib}
\usepackage{ulem}
\usepackage{color}
\usepackage{textcomp}
\usepackage{multicol}
\usepackage{multirow}

\newtheorem{thm}{Theorem}
\newtheorem{lem}{Lemma}

\newtheorem{prop}{Proposition}

\newcounter{rmk}

\definecolor{red}{RGB}{139,0,18}
\definecolor{lightred}{RGB}{186,25,31}
\definecolor{blue}{RGB}{0,56,108}
\definecolor{lightblue}{RGB}{69,100,139}
\renewcommand\emph[1]{{\color{red}\itshape #1}}

\def\t{ \mathrm{\scriptscriptstyle T} }

\def\mR{\mathds{R}}

\def\t{\mathrm{\scriptscriptstyle T} }
\def\argmin{\mathop{\arg\min}}

\title{An iterative algorithm for high-dimensional linear models with both sparse and non-sparse structures}
\author{Shun Yu and Yuehan Yang\thanks{Corresponding Author. Email: yyh@cufe.edu.cn.}\\
{\it \small School of Statistics and Mathematics, Central University of Finance and Economics,}\\
{\it \small Beijing, China}}
\date{}

\begin{document}
\maketitle
\begin{abstract}
Numerous practical medical problems often involve data that possess a combination of both sparse and non-sparse structures. Traditional penalized regularizations techniques, primarily designed for promoting sparsity, are inadequate to capture the optimal solutions in such scenarios. To address these challenges, this paper introduces a novel algorithm named Non-sparse Iteration (NSI). The NSI algorithm allows for the existence of both sparse and non-sparse structures and estimates them simultaneously and accurately. We provide theoretical guarantees that the proposed algorithm converges to the oracle solution and achieves the optimal rate for the upper bound of the $l_2$-norm error. Through simulations and practical applications, NSI consistently exhibits superior statistical performance in terms of estimation accuracy, prediction efficacy, and variable selection compared to several existing methods. The proposed method is also applied to breast cancer data, revealing repeated selection of specific genes for in-depth analysis.
\end{abstract}


\section{Introduction}
As information technology continues to advance, many fields such as finance, medicine, and environmental science require analysis of high-dimensional and large sample data. To generate easily interpretable and meaningful regression models, many regularization methods have been developed based on the assumption of sparsity. These methods include Lasso, proposed by \citet{tibshirani1996regression}, SCAD, a variable selection method based on non-concave penalized likelihood proposed by \citet{fan2001variable}, the elastic net proposed by \citet{zou2005elastic}, and the square root lasso method proposed by \citet{belloni2011square}. Additionally, \citet{hazimeh2020fast} developed a fast algorithm based on coordinate descent and local combinatorial optimization, which outperforms current sparse learning algorithms in terms of prediction, estimation, and variable selection metrics.

However, the aforementioned methods all rely on the sparsity assumption, which may not hold in many applications. Recently, the emergence of modeling dense signals has raised concerns \citep{dobriban2018high}. For instance, in the study of complex traits in biology, the signals associated with these traits tend to be distributed over most of the genome, as \citet{boyle2017expanded} proposed in their omnigenic model. This model further supports the polygenic nature of many complex traits and underscores the need that can effectively model dense signals without relying on sparsity assumptions.

Non-sparsity poses a significant challenge to traditional statistical methods and theoretical properties. Many researchers focus on the statistical inference and applications of non-sparse linear models.
For example, \citet{belloni2014inference} allowed the number of relevant variables to grow at $o(\sqrt n /\log p)$ and proposed to test univariate parameters in high-dimensional sparse models. \citet{zhu2017breaking} removed dimensionality restrictions on parametric hypothesis testing. Additionally, \citet{bradic2022testability} verified the accuracy of statistical inference without imposing model sparsity.
Furthermore, non-sparse settings are becoming increasingly prevalent. For example, \citet{zheng2021nonsparse} used principal component score vectors to obtain non-sparse effects, while \citet{zhao2022polygenic} investigated the statistical properties of polygenic risk scores in a high-dimensional but non-sparse setting. Additionally, \citet{shi2022single} considered single parameter inference of non-sparse logistic regression models. These studies demonstrate the growing interest in developing statistical methods that can effectively handle non-sparsity in high-dimensional settings.

To tackle the challenges posed by high-dimensional data in non-sparse settings and to enhance estimation accuracy, in this paper, we consider the high-dimensional linear models allowing both sparse and non-sparse structures. To accomplish this goal, we first propose a novel model that divides the regression problem into two components, the sparse part and the non-sparse part. Then we propose a novel algorithm named Non-sparse and sparse iteration (NSI) to efficiently estimate the complex model. NSI is an iterative algorithm with fixed residual variables and a global loss function. Specifically, to handle non-sparse models and allow for both sparse and non-sparse structures in high dimensional settings, we iteratively optimize the sparse part and the non-sparse part while minimizing the overall error size. This iterative approach helps to achieve improved accuracy in estimating parameters in non-sparse settings, where traditional methods may not perform well.

The proposed method has been compared to several existing methods through simulations, and the results demonstrate the superior performance of NSI in terms of statistical performance in estimation, prediction, and variable selection. Furthermore, the algorithm has been applied to a database of cancer prognostic risk scores, providing meaningful results that could potentially be used in clinical practice. NSI represents a promising new approach to analyzing high-dimensional non-sparse data, and its effectiveness has been demonstrated in both simulations and real-world applications.

The rest of this paper is organized as follows. Section 2 introduces the non-sparse iterative algorithm. In Section 3, we discuss the theoretical properties of the algorithm. Section 4 and Section 5 provide numerical simulations and applications.

\section{Model and methods}
\subsection{Model setting}
In this section, we first introduce some notations. We set $y={{({{y}_{1}},{{y}_{2}},...,{{y}_{n}})}^{T}}$ be the $n$-dimensional vector of response, $Z$ be the $n\times p$ design matrix with $p$ predictors, and $W$ be the $n\times q$ design matrix with $q$ predictors. We focus on the following high-dimensional model:
\begin{equation}\label{eq model}
y=Z\beta +W\gamma +\varepsilon,
\end{equation}
where $\beta =(\beta_1,\beta_2,\dots,\beta_p)^\t$ and $\gamma =(\gamma_1, \gamma_2,\dots,\gamma_q)^\t$ are the coefficient vectors of $Z$ and $W$, respectively, and $\varepsilon =(\varepsilon_1, \varepsilon_2, \dots, \varepsilon_n)^\t$ is the vector of random errors. Assuming that the coefficient vectors $\beta$ and $\gamma$ are sparse and non-sparse, respectively, we do not place any sparse restriction on $\gamma$. We consider a model that allows for both sparsity and non-sparsity, which is commonly observed in many applications. For example, in the study of human diseases, microbial communities associated with conditions such as obesity and inflammatory bowel disease may play a crucial role \citep{cai2023rank}. In the case of inflammatory bowel disease, the microbial communities inhabiting the bowel may be relevant, while those inhabiting other parts of the body may not be relevant. Therefore, we process the following two parts separately.

\textbf{Non-sparse part:} We leverage the information of the precision matrix to estimate the non-sparse regression coefficients. Following the idea in \cite{bradic2022testability}, we transform the challenge of estimating the non-sparse regression coefficients into estimating the corresponding non-sparse rows of the precision matrix of the covariates. Let $W$ be a set of $n$ independent and identically distributed vectors. We write $W_i$ be the $i$th sample vector of $W$ and let $\Sigma = E(W_i W_i^\t)$ be the covariance matrix of $W$. The precision matrix of $\Sigma$ is denoted by $\Omega = \Sigma^{-1}$. To estimate the coefficients, we start by considering the following function:
\begin{equation}\label{eq nonsparse}
\gamma_0  = \Omega W^\t y /n.
\end{equation}
The above function is obtained by using the definition of the true coefficients that only take into account the expectation of the dense model, that is, assume $y = W \gamma_0 + \epsilon$, and we calculate $E(y_i) = E(W_i) \gamma_0$. Based on $\Omega = \Sigma^{-1}$, we obtain \eqref{eq nonsparse}. Inspired by the above calculation, for the model \eqref{eq model}, we have
\[ E(y_i) = E (Z_i \beta) + E(W_i) \gamma, \]
thus we have,
\[EW_i^T{W_i}\gamma  = E{W_i^T}{y_i} - E{W_i^T}{Z_i}\beta ,\]
and for $\Omega = \Sigma^{-1}$,
\[\gamma = \Omega E[W^\t_i(y_i - Z_i \beta)].\]
Since $\beta$ is unknown, \eqref{eq nonsparse} is insufficient for the model \eqref{eq model} when $\beta$ has non-zero elements. To address this issue, we propose an iterative algorithm for solving $\hat{\gamma}$ and $\hat{\beta}$ iteratively. During each iteration, solving $\hat{\gamma}$ is equivalent to solving the following function:
\[ \hat \gamma = \Omega (W^\t y - Z \hat \beta). \]

\textbf{Sparse part:} Given $\hat \gamma$, the loss function of solving $\beta$ is as following,
\begin{equation}\label{eq sparse}
\hat \beta  = \argmin_\beta L(\beta, \hat \gamma; \lambda),
\end{equation}
where
\[ L(\beta, \hat \gamma; \lambda) \triangleq \dfrac{1}{2n}\|y - Z\beta  - W\hat \gamma\|^2_2 + \lambda\|\beta\|_1. \]
The loss function \eqref{eq sparse} is used as the objective function when estimating the sparse part of the model \eqref{eq model}. The sparsity is ensured by the tuning parameter $\lambda$. To provide a joint estimate for the true coefficient vectors, we propose an iterative algorithm that simultaneously recovers the two partial predictors.
In this algorithm, we first calculate the initial estimate using \eqref{eq nonsparse} and \eqref{eq sparse}. Then, during each iteration, we estimate the non-sparse and sparse parts sequentially. For the non-sparse part, we calculate the estimate using the estimated precision matrix and the partial residual. For the sparse part, we calculate the estimate using the least squares estimate of fitting the partial residual and the soft-thresholding operator. This strategy efficiently reduces the bias by constantly adjusting the target estimate with information from the other estimated coefficients. The details of the algorithm are as follows.

\begin{center}
\textbf{Algorithm: Non-sparse and sparse iteration}
\end{center}

\textbf{Step 1:} Given $\lambda $, $y$, $Z$, $W$, the initial estimate is generated as following,
\[{\hat \gamma^{[0]}} = \Omega W^\t y/n, \]
\[{{\hat \beta }^{[0]}} = \arg {\min _\beta }L(\beta,\hat \gamma^{[0]};\lambda ) = \arg {\min _\beta }\frac{{\| {y - Z\beta  - W\hat \gamma^{[0]}} \|_2^2}}{2n} + \lambda \| \beta\|_1 .\]

\textbf{Step 2:} Repeat the following steps until convergence:

For $j=1 ,2,...q$, we update $\gamma _{j}^{[t]}$ as following,
\[\hat \gamma _j^{[t]} =\frac{{{W}_j^\t ({y - \hat y}_{ - j}^{[t]} )}}{n}, ~\text{where}~ {\hat y}_{-j}^{[t]}{\rm{ = Z}}{{\hat \beta }^{[{\rm{t}}]}}{\rm{ + }}{{{W}}_{ - j}}\hat \gamma _{ - j}^{[t]}.\]

For $k=1, 2,...,p$, we update $\beta _{k}^{[t]}$ as following,
\[\beta_k^{[t]} = S (\frac{{Z_k^\t (y - {\hat y}_{- k}^{[t]} )}}{n} ,\lambda  ), ~\text{where}~{\hat y}_{-k}^{[t]}{\rm{ = }}{{{Z}}_{-k}}\hat \beta_{ - k}^{[t]}{{ + W}}{\hat \gamma ^{[t]}},\]
and $S (\alpha  ,\varsigma  ) = sign (\alpha){(\left| \alpha  \right| - \varsigma  )_ + } = \left\{ {\begin{array}{*{20}{l}}
{\alpha  - \varsigma,\quad \alpha  > {\rm{0 \ and \ }}\left| \alpha  \right| > \varsigma }\\
{\alpha  + \varsigma,\quad  {\rm{   }}\alpha  < {\rm{0 \ and\  }}\left| \alpha  \right| > \varsigma }\\
{0,\qquad \qquad \ \left| \alpha  \right| < \varsigma }.
\end{array}} \right.$

In the above algorithm, the precision matrix can be estimated by the graphical lasso \citep{friedman2008sparse}. In practice, we often use a naive approach that makes a product of the estimate $\hat \Omega$ and the rest terms to construct $\hat \gamma$. Compared to methods such as Lasso \citep{tibshirani1996lasso}, SCAD \citep{fan2001variable}, adaptive lasso \citep{zou2006adaptive}, Elastic net \citep{zou2005elastic}, and others, which are designed for sparse linear regression, NSI is specifically tailored to handle high-dimensional data with both sparse and non-sparse structures, which are commonly encountered in real-world applications. This unique feature allows NSI to simultaneously capture both dense and sparse signals, providing a more accurate estimation of the coefficients of the model.

A significant amount of research has been devoted to studying high-dimensional sparse linear models using penalty functions and iterative algorithms, such as \citet{hazimeh2020fast,mai2019iterative}. However, for problems in the high-dimensional non-sparse setting, there still exist many gaps. Previous work has focused on inference methods, as seen in \citet{zhu2017breaking,bradic2022testability}. Our algorithm attempts to fill this gap by combining penalty terms with iterative algorithms for non-sparse structures. Very few iterative algorithms have been studied in this type of structure. Experimental results on both real and synthetic datasets demonstrate that our algorithm outperforms other methods in terms of statistical performance in estimation, prediction, and variable selection. Overall, NSI offers several unique features that make it a valuable tool for high-dimensional data analysis, particularly in cases where the coefficient vector is non-sparse.

\section{Theoretical Results}

In this section, we focus on the properties of the high-dimensional linear model with a random design. Specifically, we assume that the $W \sim N(0, \Sigma)$, which is unknown. The matrix $Z$ can be either fixed or have a random design. Additionally, we introduce the noise term $\varepsilon \sim N(0,\sigma^2)$. In the high-dimensional setting, we assume that the number of samples $n$ is much smaller than the combined dimensions $p+q$. Both $q$ and $p$ are allowed to increase as the number of samples $n$ grows. We do not index them with $n$ for notational simplicity. We consider both sparse and non-sparse structures, i.e., the number of non-zero entries in the coefficient vectors, $|S| = | {i \in \{ 1,2, \cdots, p\} :{{\hat \beta }_i} \ne 0} | + | {j \in \{ 1,2, \cdots, q\} :{{\hat \gamma }_j} \ne 0} |$ is always greater than $q$ and smaller than $p + q$.

We first provide the following result to show that the proposed algorithm converges to the oracle solution.
\begin{thm}\label{thm 1}
Assume $ \Lambda_{max}(\Sigma) < \kappa_1 < \infty $. The proposed algorithm converges to the oracle solutions $\hat \beta^{\text{oracle}}$ and $\hat \gamma^{\text{oracle}}$, defined as follows:
\[ \hat \beta^{\text{oracle}} = \argmin\{ \|y - W\gamma - Z\beta\|^2_2 + \lambda\|\beta\|_1 \}, \]
and
\[ \hat \gamma^{\text{oracle}} = \Omega W^\t(y - Z\hat \beta^{\text{oracle}})/n.\]
\end{thm}
The requirement in the theorem that $ \Lambda_{max}(\Sigma) < \kappa_1 < \infty $ ensures that the proposed algorithm converges to the oracle solutions, which represent the ideal solutions in terms of accuracy and sparsity. This condition is reasonable as $\Sigma$ represents the covariance matrix of predictors in the dense structure. We use this condition to guarantee that the sample covariance matrix has an exponential type tail bound. By incorporating this condition, the algorithm can effectively handle high-dimensional problems with both sparse and non-sparse structures, leading to the following upper bound.

\begin{thm}\label{thm 2}
Assume the restrictive eigenvalue condition, which states that for a positive constant $K$,
\[v^\t \{X^\t X/n \}v \geqslant \kappa_2 \|v\|^2_2,\]
holds for all $v \in G(S)$ where $G(S) = \{ v \in \mR^p: \|v_{S^c}\|_1 \leqslant 3 \| v_S \|_1 \}$. The following error bounds for the estimate hold,
\[P(\|\hat \gamma - \gamma\|^2_2 \geqslant \delta) \leqslant o(\exp(-n\delta^2)),\]
\[ P(\|\hat \beta - \beta\|^2_2 \geqslant \delta)\leqslant o(\exp(-n\delta^2)).\]
\end{thm}

The restrictive eigenvalue condition is a common requirement in obtaining upper bounds for penalized regularizations \citep{negahban2012unified}. Specifically, when considering the precision matrix $\Omega$, if we use the Graphical Lasso as the initial estimator with tuning parameter $\lambda^*$, the following result holds.
\begin{prop}
Assume there exists some constant $ \alpha \in (0,1] $ such that
\[ \max\limits_{ e \in S^c }\|\Gamma_{e,S} (\Gamma_{SS})^{-1}\|_1 \leqslant 1-\alpha,\]
where $\|\cdot \|_1$ denotes the usual $\ell_1$ norm of a vector and $\Gamma = \Theta^{-1} \otimes \Theta^{-1}$, with $\otimes$ denoting the Kronecker matrix product. Set $\lambda^*= 4\frac{M}{\alpha}(\frac{\log (np)^{\tau}}{n})^{1/2} $, where $\tau >1$, $ M >0$. Let $ \theta_{min} $ be the minimum absolute value of nonzero entries of $ \Theta $. Assume $\theta_{min} > 4K(\Gamma)(1+\frac{4M}{\alpha}) (\frac{\tau\log (nq)}{n})^{1/2}$, where $K(\Gamma) = \| \Gamma^{-1}_{SS} \|_{\infty}$ with $\| \cdot \|_{\infty}$ denoting the usual $\ell_\infty$ norm of a matrix, e.g. $\| \Sigma \|_\infty =\max\limits_{j = 1,\ldots,p}\sum_{j'=1}^p|\Sigma_{jj'}|$.
Further, suppose the sample size satisfies
\[(\dfrac{\tau\log(nq)}{n})^{1/2} \leqslant \{ 3(1+\dfrac{4M}{\alpha}) K(\Gamma) \|\Sigma\|_{\infty}^3 l_0 \}^{-1},\]
then the following result holds with probability at least $ 1 - \dfrac{1}{p^{\tau-2}}$:
\[\|\hat \gamma - \gamma\|^2_2 \leqslant 2K(\Gamma)(1+\dfrac{4M}{\alpha}) (\dfrac{\tau(p+q)\log(p)}{n})^{1/2}.\]
Following the same arguments as Theorem \ref{thm 1} and \ref{thm 2}, we have that the proposed algorithm with the initial estimate obtained from the graphical lasso also converges to the oracle solution and achieves the optimal rate for the upper bound of $l_2$ norm error.
\end{prop}
The aforementioned result is derived from the properties of the Graphical Lasso and the proofs of Theorem \ref{thm 1} - \ref{thm 2}, thus we omit the proof and refer the readers to \citet{yyh2020two,ravikumar2011high}.

\section{Simulations}

In this section, we illustrate the performance of the proposed method and compare it with three methods: Lasso \citep{tibshirani1996regression}, Bradic \citep{bradic2022testability}, and Fast Best Subset Selection (FBS) \citep{hazimeh2020fast}. The R glmnet package \citep{friedman2010regularization} is utilized to run the lasso, and FBS is implemented using the R L0Learn package \citep{hazimeh2020fast}. Bradic is proposed for the non-sparse structure and unsuitable for the combination of the non-sparse and sparse structures. To address this problem, we propose using an adaptive version of Bradic in this paper. Specifically, we split the coefficient vector into two parts and estimate them separately. Especially, we estimate the coefficients of the non-sparse part using the Bradic; for the coefficients of the sparse part, we adopt commonly used regularization constraint methods for estimation, such as the Lasso. The results presented later for Bradic are from the adaptive version. We conduct 100 simulations for each setting.

Consider the following linear model:
\[y=Z\beta +W\gamma +\varepsilon,\]
where $\beta  = {{\rm{(}}\overbrace {4,4,...4}^{10},0,...0{\rm{)}}^\t} \in {R^p},\gamma  = {{\rm{(}}6,6,...,6{\rm{)}}^\t} \in {R^q}$,
and $\varepsilon  \sim N(0,I)$. Set $X = (Z, W)$. Given $\Omega$, we generate $X \sim N (0,\Sigma)$, where $\Sigma  = {\Omega ^{{\rm{ - 1}}}}$. We use the following two examples:

\textbf{Example 1:} We set two dimension settings $(p + q, n) = (100, 100)$ and $(p + q, n) = (400,400)$. We consider $\Omega  = I$ and different sparsity by changing the size of $q/(p + q)=0.5, 0.6, 0.7, 0.8$.

\textbf{Example 2:} Set $\Omega_{jj} = 1$ for $j = 1,\dots,(p+q)$, ${\Omega _{j,j - 1}} = {\Omega _{j - 1,j}} = \rho $ for $j = 2,\dots,(p+q)$. We fix the dimensions that $(p, q, n)= (50, 50, 100)$ and $(p, q, n)=(200, 200, 400)$. We consider the different structures of $\Omega$ by a changing $\rho=0, 0.1, 0.3, 0.5$.

To test the performance of the combinations of sparse and non-sparse structures, we vary the sparsity from $0.5$ to $0.8$, which means that the maximum number of nonzero coefficients achieves $9/10$ of the dimensions. In this case, we do not use higher dimensions in our examples. For all the methods, we use ten-fold cross-validation to select the tuning parameter. The means and standard deviations of the results are presented in Table \ref{table 1} - \ref{table 4} and consist of the following metrics.
\begin{itemize}
\item $\textbf{$l_1$-norm:} |\beta  - \hat \beta {|_1} + |\gamma  - \hat \gamma {|_1}$.
\item $\textbf{$l_2$-norm:} (|\beta  - \hat \beta |_2^2 + |\gamma  - \hat \gamma |_2^2)^{1/2}$.
\item \textbf{FPR:} False positive rate
\[\frac{{\mid j \in \{ 1,2, \cdots ,p\} :{{\hat \beta }_j} \ne 0~{\rm{ and }}~{\beta _i} = 0\mid  + \mid j \in \{ 1,2, \cdots ,q\} :{{\hat \gamma }_j} \ne 0~\rm{ and}~{\gamma _j} = 0\mid }}{{| {i \in \{ 1,2, \cdots ,p\} :{\beta _i} = 0} | + | {j \in \{ 1,2, \cdots ,q\} :{\gamma _j} = 0} |}}.\]
\item \textbf{TPR:} True positive rate
\[\frac{{\mid i \in \{ 1,2, \cdots ,p\} :{{\hat \beta }_i} \ne 0{\rm{ ~and~ }}{\beta _i} \ne 0\mid  + \mid j \in \{ 1,2, \cdots ,q\} :{{\hat \gamma }_j} \ne 0~\rm{ and }~{\gamma _j} \ne 0\mid }}{{| {i \in \{ 1,2, \cdots ,p\} :{\beta _i} \ne 0} | + | {j \in \{ 1,2, \cdots ,q\} :{\gamma_j} \ne 0} |}}.\]
\item \textbf{NZ:} The number of non-zero estimate
\[| {i \in \{ 1,2, \cdots ,p\} :{{\hat \beta }_i} \ne 0} | + | {j \in \{ 1,2, \cdots ,q\} :{{\hat \gamma }_j} \ne 0} |.\]
\end{itemize}

Table \ref{table 1} and Table \ref{table 2} present the comparison of all methods under the first example. It can be seen that our method consistently outperforms other methods in both estimation and selection. Specifically, when sparsity decreases, the proposed method accurately estimates the coefficients with lower errors and has a low standard deviation in all cases. Meanwhile, the selection result always has lower false positives. In contrast, Lasso and FBS tend to select many false positives. Moreover, in the higher dimensional case, NSI also maintains good performance as sparsity varies.

Example 2 investigates the scenario where variables are correlated. The results are presented in Table \ref{table 3} and Table \ref{table 4}. As the correlations among predictors become stronger, the estimation errors of all methods increase. NSI consistently outperforms the other methods in almost all cases of Example 2, regardless of the dimensionality. Specifically, NSI achieves the smallest prediction error across various settings, indicating its superior performance in handling data with complex correlation structures.

To further demonstrate the variation of the estimation error,  in Figure \ref{fig:method}, we compare the $l_2$-norm estimation error of the coefficients for all methods across all examples. The black line with a square symbol represents the mean of the $l_2$-norm error and the shaded region around it represents the standard deviation interval of the NSI, while the other lines with different symbols represent other methods. As shown in Figure \ref{fig:method}, in almost all cases, NSI exhibits the smallest error. In the low-dimensional case of the first example, the estimation error of NSI is significantly smaller than that of the other methods, while in the high-dimensional case, the increase in the errors of Bradic and NSI is more gradual, while the other two methods perform worse. Moreover, NSI maintains a more stable low error and standard deviation as the correlation between variables increases.

\begin{table}[!htp]
\centering
\small
\setlength\tabcolsep{7pt}
\caption{Performance comparison under example 1 ($p+q=100$)}
\label{table 1}
\scalebox{0.9}{
\begin{tabular}{ccccccc}
\hline
&Method&    $l_2$-norm&$l_1$-norm&FPR&TPR&NZ\\
\hline
$q/p+q=0.5$&    Lasso   &   13.418(5.458)   &   97.175(41.164)  &   0.549(0.102)    &   0.992(0.019)    &   81.43(3.783)    \\\cline{2-7}
NZ$=$60&Bradic  &   15.531(1.532)   &   104.671(9.877)  &   0.354(0.152)    &   0.99(0.017) &   73.58(6.609)    \\\cline{2-7}
&FBS    &   14.924(7.811)   &   104.327(57.092) &   0.263(0.117)    &   0.941(0.073)    &   66.95(2.879)    \\\cline{2-7}
&NSI    &   10.947(2.018)   &   67.395(13.566)  &   0.046(0.038)    &   0.988(0.016)    &   61.14(1.457)    \\\cline{2-7}
\hline
$q/p+q=0.6$
&Lasso  &   21.43(5.786)    &   161.782(46.491) &   0.597(0.115)    &   0.964(0.053)    &   85.37(5.624)    \\\cline{2-7}
NZ$=$70 &Bradic &   18.701(2.078)   &   133.847(14.929) &   0.376(0.194)    &   0.983(0.022)    &   80.12(6.7)  \\\cline{2-7}
&FBS    &   17.076(6.372)   &   126.676(48.769) &   0.383(0.1)  &   0.95(0.059) &   78(2.828)   \\\cline{2-7}
&NSI    &   12.344(2.116)   &   82.838(15.011)  &   0.076(0.051)    &   0.988(0.014)    &   71.43(1.526)    \\\cline{2-7}
\hline
$q/p+q=0.7$
&Lasso  &   26.889(4.178)   &   210.277(37.051) &   0.618(0.104)    &   0.94(0.045) &   87.57(4.13) \\\cline{2-7}
NZ$=$80&Bradic  &   22.479(2.88)    &   168.947(22.604) &   0.354(0.235)    &   0.971(0.029)    &   84.76(6.512)    \\\cline{2-7}
&FBS    &   22.293(2.802)   &   177.35(23.477)  &   0.547(0.108)    &   0.963(0.026)    &   88(1.98)    \\\cline{2-7}
&NSI    &   14.939(2.873)   &   108.005(22.546) &   0.098(0.071)    &   0.98(0.015) &   80.33(1.621)    \\\cline{2-7}
\hline
$q/p+q=0.8$
&Lasso  &   31.716(3.396)   &   256.766(31.817) &   0.646(0.17) &   0.91(0.048) &   88.37(4.911)    \\\cline{2-7}
NZ$=$90&Bradic  &   28.813(4.52)    &   226.294(35.451) &   0.363(0.299)    &   0.958(0.033)    &   89.83(5.527)    \\\cline{2-7}
&FBS    &   25.42(2.582)    &   202.318(23.739) &   0.501(0.162)    &   0.922(0.028)    &   88(2.079)   \\\cline{2-7}
&NSI    &   25.402(2.855)   &   197.596(23.915) &   0(0)    &   0.889(0)    &   80(0)   \\\cline{2-7}
\hline
\end{tabular}}
\end{table}

\begin{table}[!htp]
\centering
\small
\setlength\tabcolsep{6.5pt}
\caption{Performance comparison under example 1 ($p+q=400$)}
\label{table 2}
\scalebox{0.9}{
\begin{tabular}{ccccccc}
\hline
&   Method  &   $l_2$-norm  &   $l_1$-norm  &   FPR &   TPR &   NZ  \\
\hline
$q/p+q=0.5$
&   Lasso   &   15.772(3.634)   &   214.151(51.994) &   0.454(0.055)    &   1(0)    &   296.3(10.421)   \\\cline{2-7}
NZ $=$ 210  &   Bradic  &   14.894(0.826)   &   171.073(10.446) &   0.139(0.069)    &   1(0)    &   236.4(13.054)   \\\cline{2-7}
&   FBS &   13.937(12.638)  &   175.168(166.723)    &   0.061(0.085)    &   0.981(0.05) &   217.76(8.815)   \\\cline{2-7}
&   NSI &   10.271(1.139)   &   104.477(12.37)  &   0(0)    &   1(0)    &   210(0)  \\\cline{2-7}
\hline
$q/p+q=0.6$
&   Lasso   &   33.003(6.351)   &   477.789(96.258) &   0.568(0.054)    &   0.996(0.007)    &   334.1(7.368)    \\\cline{2-7}
NZ $=$ 250  &   Bradic  &   17.891(1.131)   &   223.333(15.448) &   0.169(0.062)    &   1(0)    &   275.32(9.237)   \\\cline{2-7}
&   FBS &   35.189(14.285)  &   484.134(197.894)    &   0.255(0.065)    &   0.933(0.059)    &   271.44(7.976)   \\\cline{2-7}
&   NSI &   15.861(1.493)   &   183.061(19.16)  &   0(0.001)    &   0.994(0.005)    &   248.56(1.328)   \\\cline{2-7}
\hline
$q/p+q=0.7$
&   Lasso   &   48.448(5.04)    &   734.187(83.283) &   0.613(0.049)    &   0.97(0.016) &   348.65(5.342)   \\\cline{2-7}
NZ $=$ 290  &   Bradic  &   21.865(1.44)    &   294.268(21.766) &   0.196(0.081)    &   1(0.001)    &   311.47(8.933)   \\\cline{2-7}
&   FBS &   31.953(2.606)   &   460.141(43.611) &   0.338(0.047)    &   0.996(0.006)    &   273.27(8.525)   \\\cline{2-7}
&   NSI &   21.71(1.445)    &   279.812(22.314) &   0(0)    &   0.977(0.006)    &   283.28(1.652)   \\\cline{2-7}
\hline
$q/p+q=0.8$
&   Lasso   &   59.236(3.554)   &   933.473(64.229) &   0.628(0.061)    &   0.939(0.018)    &   353.82(6.851)   \\\cline{2-7}
NZ $=$ 330  &   Bradic  &   27.746(1.992)   &   400.064(31.959) &   0.24(0.094) &   0.999(0.002)    &   346.45(6.819)   \\\cline{2-7}
&   FBS &   34.028(5.504)   &   522.845(83.916) &   0.434(0.063)    &   0.986(0.015)    &   355.94(4.343)   \\\cline{2-7}
&   NSI &   26.188(1.907)   &   370.355(30.524) &   0(0.002)    &   0.982(0.005)    &   323.94(1.722)   \\\cline{2-7}
\hline
\end{tabular}}
\end{table}

\begin{table}[!htp]
\centering
\small
\setlength\tabcolsep{12pt}
\caption{Performance comparison under example 2 ($p+q=100$)}
\label{table 3}
\scalebox{0.8}{
\begin{tabular}{ccccccc}
\hline
&   Method  &   $l_2$-norm  &   $l_1$-norm  &   FPR &   TPR &   NZ  \\
\hline
$\rho=0$    &   Lasso   &   13.418(5.458)   &   97.175(41.164)  &   0.549(0.102)    &   0.992(0.019)    &   81.43(3.783)    \\\cline{2-7}
&   Bradic  &   15.531(1.532)   &   104.671(9.877)  &   0.354(0.152)    &   0.99(0.017) &   73.58(6.609)    \\\cline{2-7}
&   FBS &   14.924(7.811)   &   104.327(57.092) &   0.263(0.117)    &   0.941(0.073)    &   66.95(2.879)    \\\cline{2-7}
&   NSI &   10.947(2.018)   &   67.395(13.566)  &   0.046(0.038)    &   0.988(0.016)    &   61.14(1.457)    \\\cline{2-7}
\hline
$\rho=0.1$  &   Lasso   &   13.452(5.846)   &   97.194(43.278)  &   0.567(0.097)    &   0.987(0.039)    &   81.91(4.508)    \\\cline{2-7}
&   Bradic  &   14.738(1.635)   &   100(11.537) &   0.374(0.17) &   0.993(0.013)    &   74.5(7.113) \\\cline{2-7}
&   FBS &   18.042(7.299)   &   125.81(53.722)  &   0.293(0.119)    &   0.915(0.072)    &   66.64(3.014)    \\\cline{2-7}
&   NSI &   11.443(1.977)   &   69.987(13.462)  &   0.052(0.039)    &   0.983(0.018)    &   61.01(1.345)    \\\cline{2-7}
\hline
$\rho=0.3$  &   Lasso   &   17.503(6.675)   &   129.369(52.184) &   0.631(0.098)    &   0.976(0.058)    &   83.78(5.573)    \\\cline{2-7}
&   Bradic  &   14.144(1.609)   &   94.582(10.097)  &   0.426(0.175)    &   0.987(0.025)    &   76.26(7.998)    \\\cline{2-7}
&   FBS &   25.565(2.74)    &   191.803(23.959) &   0.384(0.071)    &   0.858(0.046)    &   56.3(3.093) \\\cline{2-7}
&   NSI &   13.256(1.562)   &   80.706(10.409)  &   0.098(0.047)    &   0.955(0.028)    &   58.38(2.415)    \\\cline{2-7}
\hline
$\rho=0.4$  &   Lasso   &   33.159(2.995)   &   238.2(27.548)   &   0.389(0.101)    &   0.725(0.111)    &   59.04(9.476)    \\\cline{2-7}
&   Bradic  &   25.93(4.086)    &   165.96(25.042)  &   0.266(0.157)    &   0.965(0.034)    &   68.54(7.926)    \\\cline{2-7}
&   FBS &   30.065(2.497)   &   208.572(24.856) &   0.364(0.058)    &   0.732(0.046)    &   58.46(1.823)    \\\cline{2-7}
&   NSI &   24.757(3.075)   &   158.425(20.178) &   0.01(0.018) &   0.855(0.013)    &   51.69(1.161)    \\\cline{2-7}

\hline
\end{tabular}}
\end{table}

\begin{table}[!htp]
\centering
\small
\setlength\tabcolsep{10.5pt}
\caption{Performance comparison under example 2 ($p+q=400$)}
\label{table 4}
\scalebox{0.8}{
\begin{tabular}{ccccccc}
\hline
&   Method  &   $l_2$-norm  &   $l_1$-norm  &   FPR &   TPR &   NZ  \\
\hline
$\rho=0$    &   Lasso   &   15.772(3.634)   &   214.151(51.994) &   0.454(0.055)    &   1(0)    &   296.3(10.421)   \\\cline{2-7}
&   Bradic  &   14.894(0.826)   &   171.073(10.446) &   0.139(0.069)    &   1(0)    &   236.4(13.054)   \\\cline{2-7}
&   FBS &   13.937(12.638)  &   175.168(166.723)    &   0.061(0.085)    &   0.981(0.05) &   217.76(8.815)   \\\cline{2-7}
&   NSI &   10.271(1.139)   &   104.477(12.37)  &   0(0)    &   1(0)    &   210(0)  \\\cline{2-7}
\hline
$\rho=0.1$  &   Lasso   &   18.493(4.791)   &   254.584(69.482) &   0.504(0.058)    &   1(0.002)    &   305.69(10.902)  \\\cline{2-7}
&   Bradic  &   14.291(0.866)   &   162.385(11.005) &   0.16(0.053) &   1(0)    &   240.31(9.981)   \\\cline{2-7}
&   FBS &   25.323(16.889)  &   325.83(223.077) &   0.141(0.102)    &   0.945(0.071)    &   225.35(8.984)   \\\cline{2-7}
&   NSI &   11.888(1.211)   &   118.558(13.419) &   0(0.001)    &   0.999(0.002)    &   209.88(0.409)   \\\cline{2-7}
\hline
$\rho=0.3$  &   Lasso   &   25.972(6.234)   &   367.844(91.737) &   0.605(0.051)    &   0.999(0.004)    &   324.65(9.396)   \\\cline{2-7}
&   Bradic  &   13.248(0.841)   &   148.489(10.048) &   0.218(0.059)    &   1(0)    &   251.35(11.282)  \\\cline{2-7}
&   FBS &   61.781(1.619)   &   934.185(31.349) &   0.346(0.033)    &   0.774(0.028)    &   228.37(6.817)   \\\cline{2-7}
&   NSI &   16.359(0.481)   &   156.932(9.02)   &   0(0)    &   0.952(0)    &   200(0)  \\\cline{2-7}
\hline
$\rho=0.5$  &   Lasso   &   76.775(1.335)   &   1077.072(23.661)    &   0.214(0.03) &   0.417(0.044)    &   128.17(13.349)  \\\cline{2-7}
&   Bradic  &   49.973(4.115)   &   568.227(45.253) &   0.133(0.045)    &   1(0.002)    &   235.24(8.716)   \\\cline{2-7}
&   FBS &   63.641(2.063)   &   906.626(34.707) &   0.38(0.036) &   0.682(0.042)    &   215.28(11.908)  \\\cline{2-7}
&   NSI &   46.382(3.818)   &   524.155(41.784) &   0.045(0.017)    &   0.996(0.004)    &   217.7(3.359)    \\\cline{2-7}

\hline
\end{tabular}}
\end{table}

\begin{figure}[ht]
\centering
\subfigure[Example 1 ($p+q=100$)]{
\includegraphics[width=.48\textwidth,height=.38\columnwidth]
{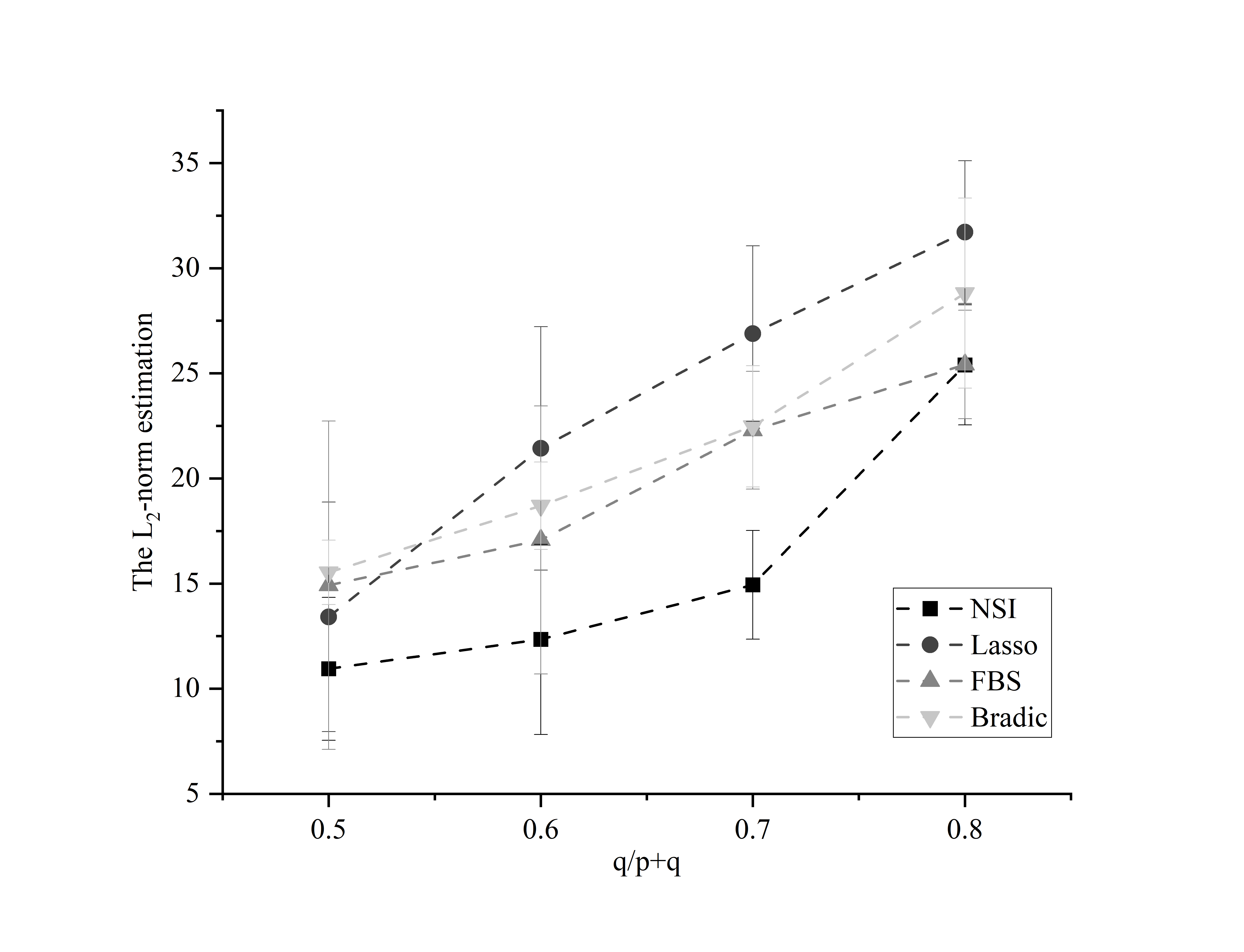}}
\subfigure[Example 1 ($p+q=400$)]{
\includegraphics[width=.48\textwidth,height=.38\columnwidth]
{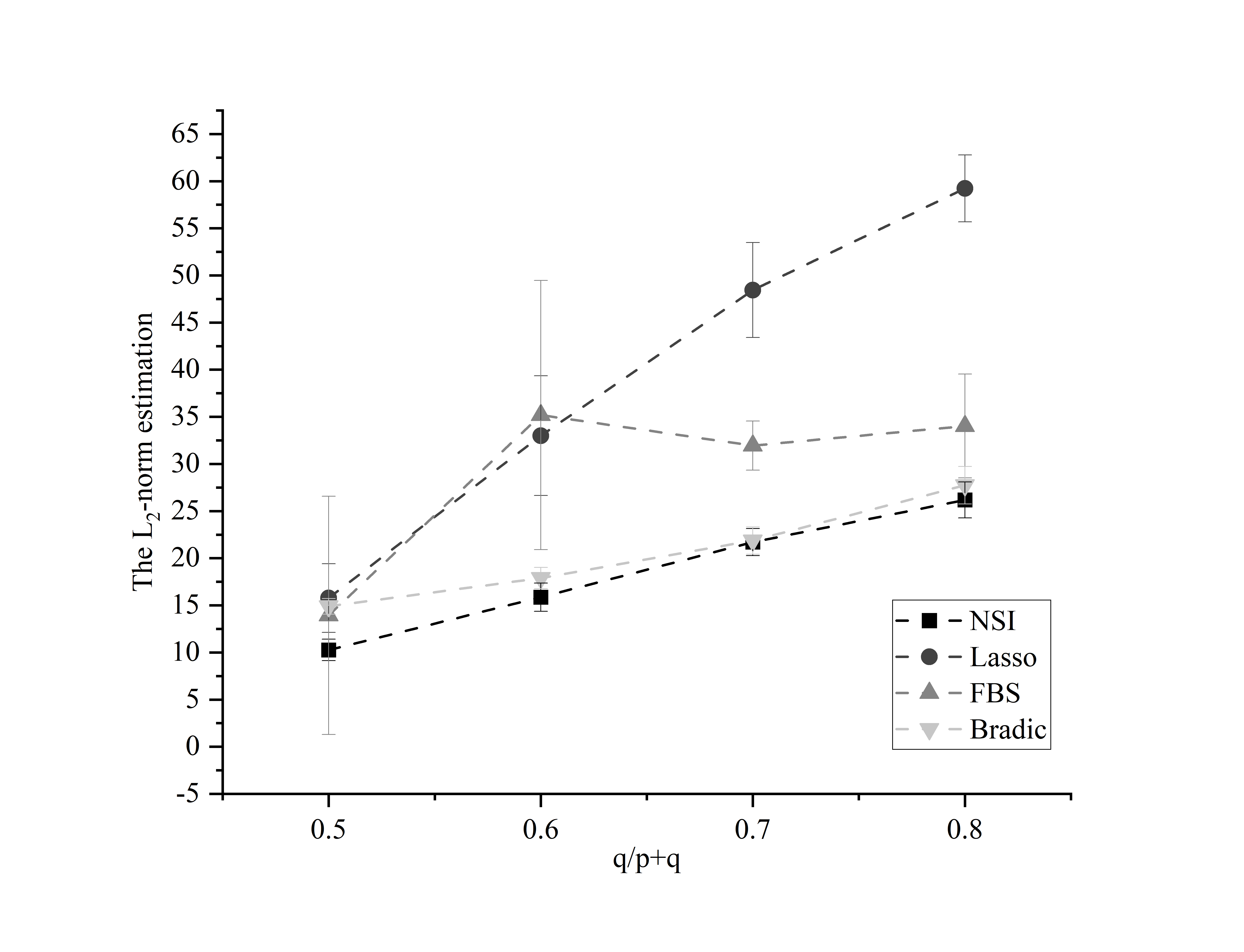}}
\subfigure[Example 2 ($p+q=100$)]{
\includegraphics[width=.48\textwidth,height=.38\columnwidth]
{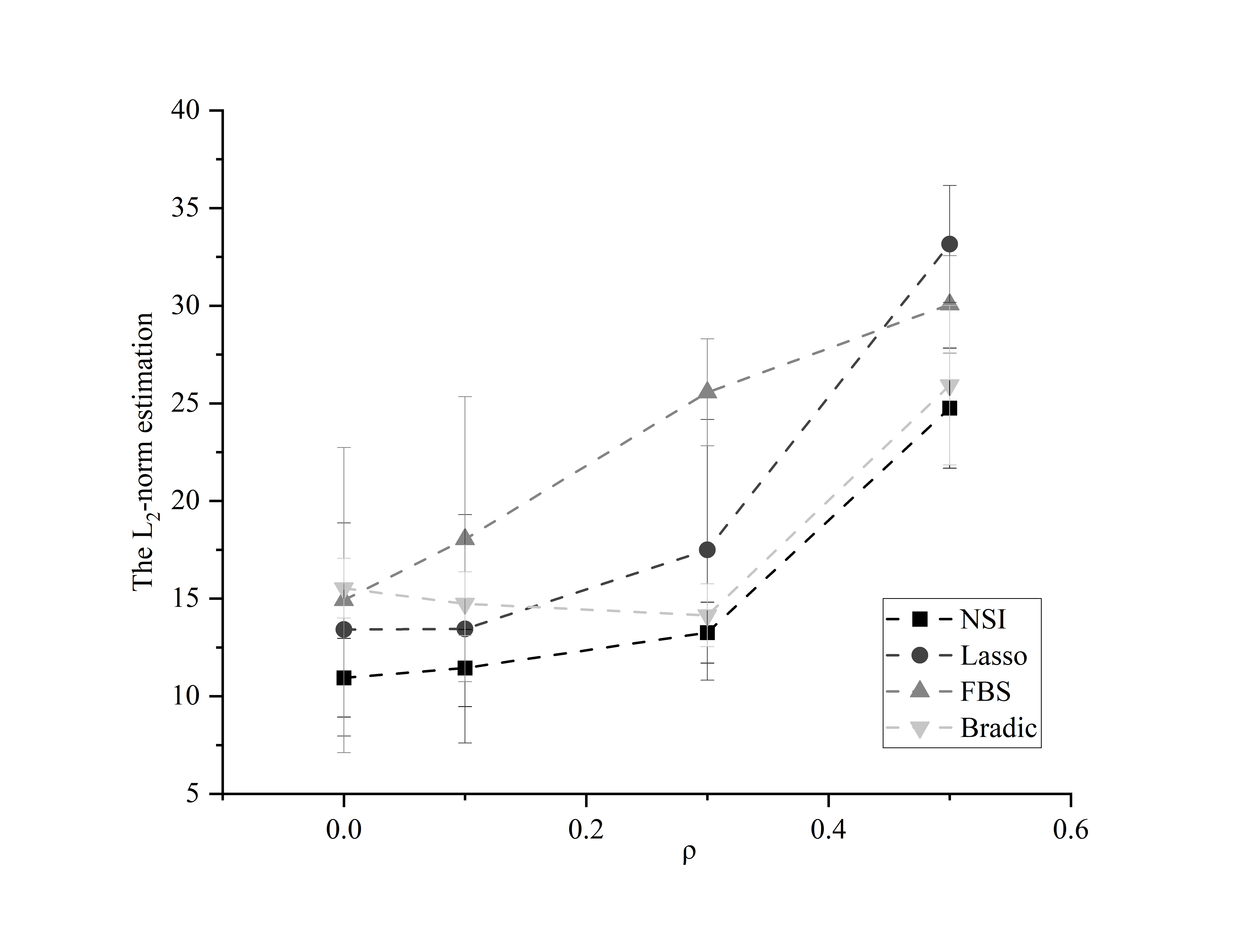}}
\subfigure[Example 2 ($p+q=400$)]{
\includegraphics[width=.48\textwidth,height=.38\columnwidth]
{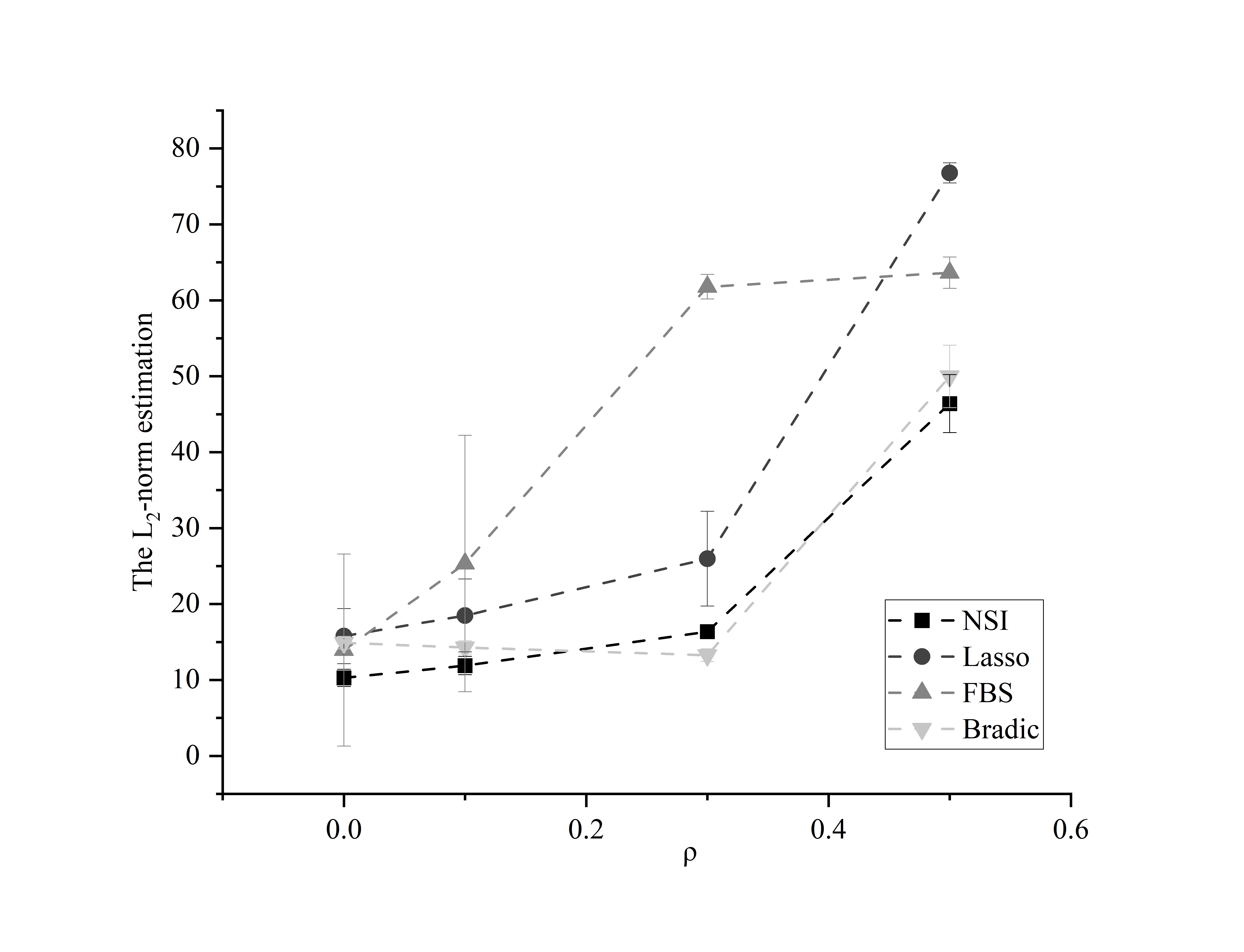}}
\caption{The $l_2$-norm estimation error: $(|\beta  - \hat \beta |_2^2 + |\gamma  - \hat \gamma |_2^2)^{1/2}$.}
\label{fig:method}
\end{figure}
\newpage
\section{Empirical Analysis}
In this section, we apply the proposed method to the expression data from invasive breast cancer patients. The data is from the GEO database\footnote{https://www.ncbi.nlm.nih.gov/geo/query/acc.cgi?acc=GSE102484}, provided by \citet{cheng2017validation}. They develop a classifier called 18-GC and quantify the prognostic risk of distant breast cancer metastasis to predict favorable or unfavorable prognosis of distant metastasis in the general breast cancer patient population. The data provides a total of 683 samples which are selected from a randomly selected breast cancer population at an independent cancer center, 54675 mRNAs extracted from frozen fresh tissue \citep{cheng2006genomic}, and the corresponding 18-GC score results. 18-GC can predict both local/regional recurrence and distant metastasis and will help enter a new era of precision medicine and breast cancer treatment \citep{cheng2017validation}.

We analyze the data by NSI and other methods, i.e., Bradic, FBS, and Lasso. Since we can not use FPR, TPR, and $l_2$-norm to measure the performance of the model, we use mean squared error (MSE) instead for comparison. We evaluate performance through the following operation. 70\% of the dataset is taken as the test set and the remaining 30\% is used as the validation set to predict the test data using the estimated model and to calculate the prediction error, and the results are in Table \ref{tablemse}. Also, the dataset involves a large number of genes while many of which have an insufficient effect on the 18-GC score. We calculate the correlation between each gene with the 18-GC score and we also explore the selection of genes and the effects by different methods under different thresholds $(0.4, 0.45, 0.5, 0.55, 0.6)$.
Figure \ref{fig:empirical} shows the MSE of the methods as well as the number of separately selected genes and commonly selected genes under different thresholds. In almost all cases, the prediction error of NSI is the lowest. At a threshold of 0.55, the prediction error is relatively lowest for almost all methods, and the number of genes commonly selected by the four methods is the highest at this time, indicating that the variables associated with the dependent variable are concentrated in the range of genes selected at this threshold.
\begin{figure}[ht]
\centering
\subfigure[MSE]{
\includegraphics[width=.48\textwidth,height=.38\columnwidth]
{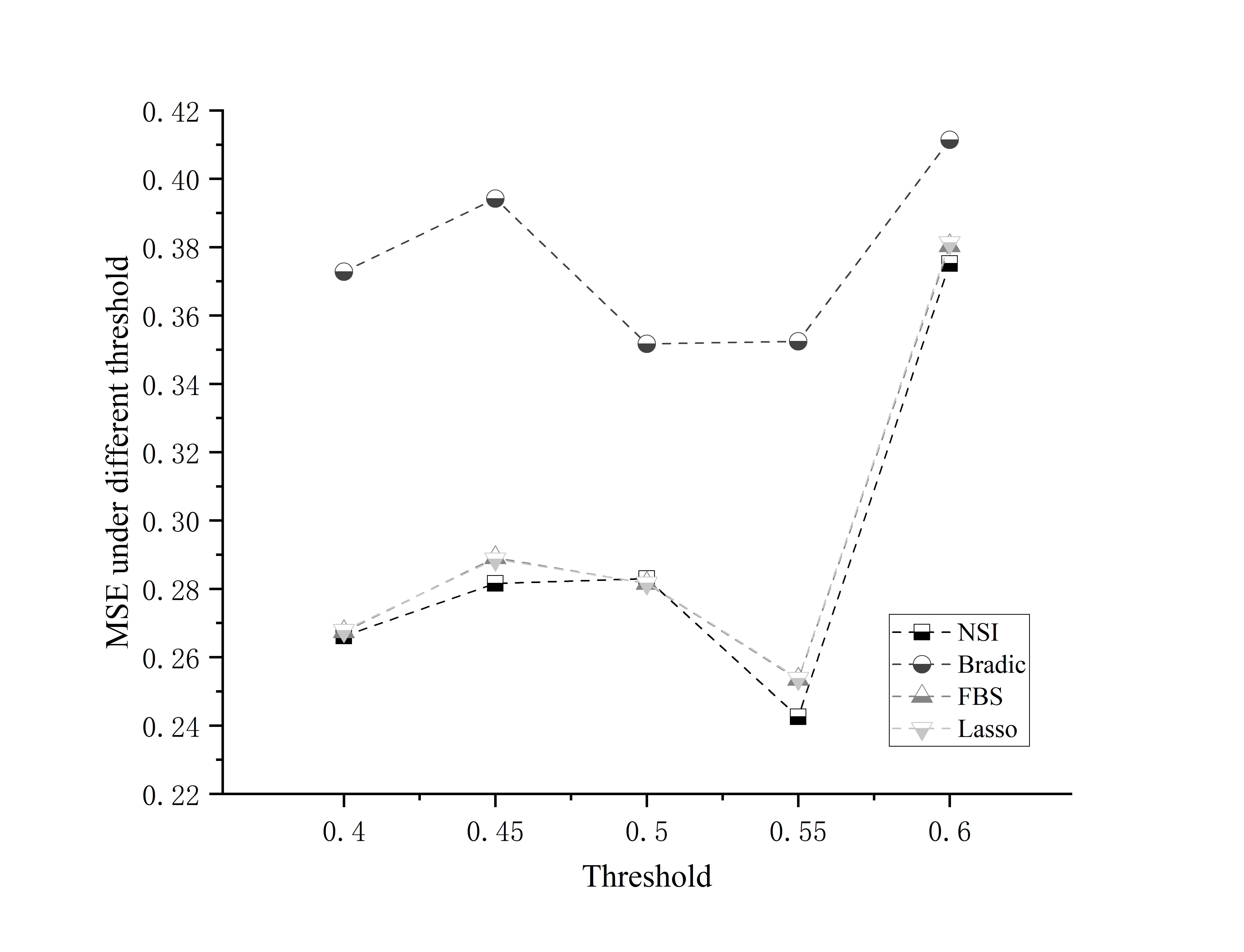}}
\subfigure[Number of selected genes]{
\includegraphics[width=.48\textwidth,height=.38\columnwidth]
{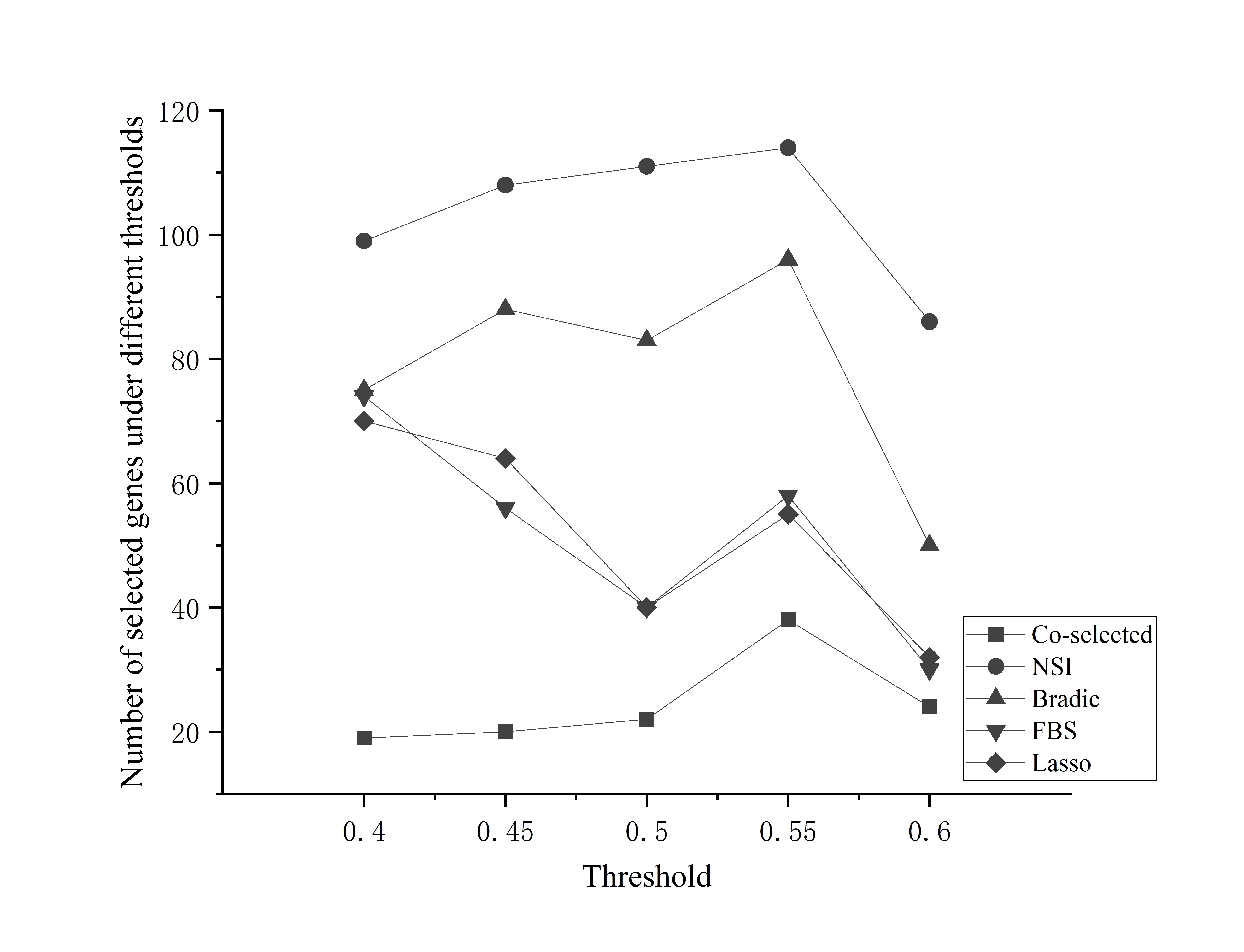}}
\caption{Performance of the methods under different thresholds}
\label{fig:empirical}
\end{figure}

\begin{table}[!htp]
\centering
\small
\setlength\tabcolsep{12pt}
\caption{MSE under different thresholds}
\label{tablemse}
\scalebox{0.8}{
\begin{tabular}{ccccc}
\hline
Threshold &NSI& Bradic& FBS&Lasso\\
\hline
0.6&    0.3752& 0.4114& 0.3804& 0.3813
\\
0.55&   0.2425& 0.3524& 0.2535& 0.2538
\\
0.5&    0.2830& 0.3517& 0.2816& 0.2817
\\
0.45&   0.2816& 0.3942& 0.2891& 0.2887
\\
0.4 &0.2662 &0.3728&    0.2674& 0.2678
\\
\hline
\end{tabular}}
\end{table}

Table \ref{tablenumber} shows the number of genes remaining under different thresholds, the number of common genes selected by all the methods, and the number of genes that are selected by each method. Several genes are repeatedly selected at different thresholds, i.e., Fanca, Rnf213, E2f8, Ddx39a, Bub1b, Rrm2, Fanci, and Tacc3. They are all known to be important factors associated with breast cancer. Specifically, Fanca and Fanci are two genes with a high frequency of deleterious mutations detected in breast cancer patients \citep{xiao2021characterization}. In breast cancer, the endocrine therapeutic response can be predicted by measuring the expression levels of specific circRNAs, and DDX39a is involved in circRNA export \citep{kristensen2022emerging}. E2F8 is significantly elevated in breast cancer cell lines and clinical breast cancer tissue samples, respectively \citep{ye2016upregulation}. Bub1b kinase can be mutated in different human malignancies, including hematopoietic, colorectal, lung, and breast cancers \citep{dai2004slippage}. Rrm2 can also be used as a biomarker for breast cancer prognosis \citep{shi2022high}. The expression levels of TACC3 mRNA and protein are higher in breast cancer tissues than in paraneoplastic tissues \citep{jiang2016clinical}.

\begin{table}[!htp]
\centering
\small
\setlength\tabcolsep{12pt}
\caption{Number of selected genes under different thresholds}
\label{tablenumber}
\scalebox{0.8}{
\begin{tabular}{ccccccc}
\hline
Threshold&  Remaining genes&Commonly selected genes &NSI&   Bradic& FBS&Lasso\\
\hline
0.6&    96&24&  86& 50& 30& 32\\
0.55&   178&38& 114&96& 58& 55\\
0.5&    483&22& 111&    83& 40& 40\\
0.45&   904&20& 108&88& 56& 64\\
0.4 &1934&19&   99& 75& 74& 70\\
\hline
\end{tabular}}
\end{table}

\section{Summary}
In this paper, we consider the problem of estimating high-dimensional data containing both non-sparse and sparse structures. To achieve high accuracy, we propose a novel iterative algorithm called Non-sparse Iteration (NSI). This algorithm is designed to perform prediction and variable selection when dealing with coefficient vectors that contain both non-sparse and sparse components. The NSI algorithm performs separate estimations for the non-sparse and sparse parts of the coefficient vector. During each iteration, we update the estimate of each coefficient by taking into account the information from other estimated coefficients. This approach is similar to coordinate descent and effectively reduces the estimation error while maintaining computational efficiency. We show that the proposed algorithm converges to the oracle solution and achieves the optimal rate for the error bounds. We evaluate the methods through simulations and real data sets. The simulations show that NSI exhibits strong performance in both estimation and selection. This method estimates the model accurately with low error and low FPR rate and still performs well both when the number of relevant variables becomes large or when the correlations between variables become large.

NSI also performs well on real data, such as genetic data analysis. In genetic data, samples are often difficult to obtain, thus genetic data usually present a situation where the number of samples is much smaller than the number of variables while the number of important variables is also large. Further, in the case of gene expression analysis, which is usually a complex project involving many factors that interact with each other, single sparse models are no longer fittable to guarantee obtaining more accurate results. As analyzed in our empirical study, our proposed iterative algorithm can obtain smaller prediction errors and more accurate results.

For the problem studied in our paper, we have obtained relevant results and demonstrated the advantages of our approach through numerical simulations as well as empirical studies. There are still many directions worth studying. For example, the problem of hypothesis testing in non-sparse learning, how to generalize to non-linear models, and analyze the identifiability of results under non-sparse models \citep{dukes2021inference, toda2023comparison, tang2023bias}. These issues are interesting topics for future research.

\section*{Acknowledgments}
This work was supported by the National Natural Science Foundation of China (Grant No. 12371281); the Emerging Interdisciplinary Project, Program for Innovation Research, and the Disciplinary Funds of Central University of Finance and Economics.

\section*{Declaration of conflicting interests}
The author declared no potential conflicts of interest with respect to the research, authorship, and/or publication of this article.


\appendix
\section*{Appendix}
To prove Theorem \ref{thm 1}, we introduce the following Lemmas
\begin{lem}\label{lem1}
Let $ Z_i $'s be independent identically distributed from $ \mathcal{N}(0,\Sigma) $ and $ \Lambda_{max}(\Sigma) < \kappa_1 < \infty $. Then for $ \Sigma = \{\sigma_{jj'}\} $, the associated sample covariance $ \hat \Sigma =\{ \hat \sigma_{jj'} \} $ satisfies the tail bound
\[
P(|\hat \sigma_{jj'} - \sigma_{jj'}|> \delta) \leqslant K_1 \exp(-K_2n\delta^2),
\]
where $ K_1 $ and $ K_2 $ are positive constants depending on the maximum eigenvalue of $ \Sigma $, and $ |\delta| \leqslant \kappa_3 $ where $ \kappa_3 $ depends on $ \kappa_1 $ as well.
\end{lem}
\begin{proof}
Detailed proof can be found in Lemma 3 of \citet{bickel2008regularized} and Lemma 1 of \citet{rothman2008sparse}. We omit proof and refer the readers to reference.
\end{proof}

\begin{proof}[Proofs of Theorem \ref{thm 1} and Theorem \ref{thm 2}]
Consider the model,
\[ y=Z\beta +W\gamma +\varepsilon. \]
For
$$\gamma = \Omega E[W^\t_i(y_i - Z_i \beta)]$$ and $$\widetilde\gamma = \Omega W^\t(y - Z \beta)/n,$$
we have that,
\begin{align*}
\widetilde\gamma - \gamma & = \Omega W^\t(W\gamma + \varepsilon)/n -\gamma\\
& = (\Omega W^\t W/n - I)\gamma + \Omega W^\t W\varepsilon/n
\sim \varepsilon,
\end{align*}
where the last term holds with exponential probability based on the condition of the noise term and Lemma~\ref{lem1}. Thus, we have that, if we use the true sparse coefficient vector, the following holds,
\[P(\|\widetilde\gamma - \gamma\|^2_2 \geqslant \delta) \leqslant o(\exp(-n\delta^2)).\]
Consider the oracle solution, such as
\[ \hat \beta^{\text{oracle}} = \argmin\{ \|y - W\gamma - Z\beta\|^2_2 + \lambda\|\beta\|_1 \}. \]
For the above oracle lasso solution, we have that, under the restricted eigenvalue condition, the following inequality holds with exponential probability too.
\[\|\hat \beta^{\text{oracle}} - \beta\|_2 \leqslant \dfrac{8\sigma}{\kappa_2} \sqrt{\dfrac{p_1\log p}{n}}. \]
Following the same notation above and set $\lambda = 4\sigma \sqrt{\log p / n}$, we have that
\[ P(\|\hat \beta^{\text{oracle}}  - \beta\|^2_2 \geqslant \delta)\leqslant o(\exp(-n\delta^2)).\]
We omit the proof of the above result and refer the readers to \citep{negahban2012unified}. In this case, set
\[ \hat \gamma^{\text{oracle}} = \Omega W^\t(y - Z\hat \beta^{\text{oracle}})/n. \]
Based on the above result, we have $$P(\|\hat \gamma^{\text{oracle}} - \gamma\|^2_2 \geqslant \delta) \leqslant o(\exp(-n\delta^2)).$$ Now we need to prove that in the iteration, $\hat \beta$ and $\hat \gamma$ converge to $\hat \beta^{\text{oracle}}$ and $\hat \gamma^{\text{oracle}}$ respectively. We use coordinate descent during the algorithm. \citet{tseng2001convergence} proved that the coordinate descent converges to the minimizer of
\[ f(\beta_1,\dots,\beta_p) = g(\beta_1,\dots,\beta_p) + \sum^p_{j= 1} h_j(\beta_j). \]
The key to the convergence is the separability of the penalty function $\sum^p_{j= 1} h_j(\beta_j)$. During the iteration, we have that both
\[\hat \gamma _j^{[t]} =\frac{{{W}_j^\t ({y - \hat y}_{ - j}^{[t]} )}}{n}, ~\text{where}~ {\hat y}_{-j}^{[t]}{\rm{ = Z}}{{\hat \beta }^{[{\rm{t}}]}}{\rm{ + }}{{{W}}_{ - j}}\hat \gamma _{ - j}^{[t]}\]
and
\[\beta _j^{[t]} = S (\frac{{Z_j^\t (y - {\hat y}_{ - j}^{[t]} )}}{n} ,\lambda  ), ~\text{where}~{\hat y}_{ - j}^{[t]}{\rm{ = }}{{{Z}}_{ - j}}\hat \beta _{ - j}^{[t]}{{ + W}}{\hat \gamma ^{[t]}}\]
are separable. Following the same arguments as the proof of \citet{tseng2001convergence}, we obtain that the proposed algorithm converges to the optimal solutions, such as the oracle solution $\hat \beta^{\text{oracle}}$ and $\hat \gamma^{\text{oracle}}$.
\end{proof}

\end{document}